\documentclass[submission,copyright,creativecommons]{eptcs}
\providecommand{\event}{FICS 2015} 

\usepackage{xspace}
\usepackage[amsthm]{ntheorem}
\usepackage{amsmath,amssymb}
\usepackage{stmaryrd}
\usepackage{latexsym}
\usepackage{txfonts}
\usepackage{tikz}
\usetikzlibrary{calc,automata,decorations.pathmorphing,shapes}
\allowdisplaybreaks

\newcommand{\mucalc}[1]{\ensuremath{\mathcal{L}_\mu^{#1}}}

\newcommand{\Mudiam}[2]{\ensuremath{\langle #1 \rangle_{#2}}}
\newcommand{\Mubox}[2]{\ensuremath{[ #1 ]_{#2}}}
\newcommand{\Mudiamanon}[1]{\ensuremath{\Diamond_{#1}}}
\newcommand{\Muboxanon}[1]{\ensuremath{\Box_{#1}}}
\newcommand{\repl}[1]{\ensuremath{\{#1\}}}

\newcommand{\sem}[3]{\ensuremath{[\![#1]\!]_{#2}^{#3}}}
\newcommand{\fp}[1]{\ensuremath{\mathit{fp}_{#1}}}
\newcommand{\ad}[2]{\ensuremath{\mathit{ad}_{#1}(#2)}}
\newcommand{\ar}[1]{\ensuremath{\mathit{ar}(#1)}}
\newcommand{\sub}[1]{\ensuremath{\mathit{Sub}(#1)}}

\newcommand{\verifier}{\textsc{Verifier}\xspace}
\newcommand{\refuter}{\textsc{Refuter}\xspace}

\newcommand{\Nat}{\ensuremath{\mathbb{N}}}
\def\newarrow#1{\mathop{{\hbox{\setbox0=\hbox{$\scriptstyle{#1\quad}$}{$%
\mathrel{\mathop{\setbox1=\hbox to
\wd0{\rightarrowfill}\ht1=3pt\dp1=-2pt\box1}\limits^{#1}}%
$}}}}}

\newcommand{\Transition}[3]{\ensuremath{#1 \newarrow{#2} #3}}
\newcommand{\Prop}{\ensuremath{\mathsf{Prop}}}
\newcommand{\Act}{\ensuremath{\mathsf{Act}}}
\newcommand{\Var}{\ensuremath{\mathsf{Var}}}
\newcommand{\lts}{\ensuremath{\mathcal{T}}}

\newcommand{\propAnd}{\ensuremath{p^{\wedge}}}
\newcommand{\propOr}{\ensuremath{p^{\vee}}}
\newcommand{\propDiam}[1]{\ensuremath{p^{\Diamond}_{#1}}}
\newcommand{\propBox}[1]{\ensuremath{p^{\Box}_{#1}}}
\newcommand{\propDiamNo}{\ensuremath{p^{\Diamond}}}
\newcommand{\propBoxNo}{\ensuremath{p^{\Box}}}
\newcommand{\propFP}[1]{\ensuremath{p^{\mathsf{FP}}_{#1}}}
\newcommand{\propFPNo}{\ensuremath{p^{\mathsf{FP}}}}
\newcommand{\propRepl}[1]{\ensuremath{p^{\mathsf{rp}}_{#1}}}
\newcommand{\propReplNo}{\ensuremath{p^{\mathsf{rp}}}}
\newcommand{\propSwapNo}{\ensuremath{p^{\mathsf{sw}}}}
\newcommand{\propStop}{\ensuremath{p^{\bullet}}}
\newcommand{\propPLit}[2]{\ensuremath{p^+_{#1,#2}}}
\newcommand{\propNLit}[2]{\ensuremath{p^-_{#1,#2}}}
\newcommand{\propPLitNo}{\ensuremath{p^+}}
\newcommand{\propNLitNo}{\ensuremath{p^-}}

\theoremstyle{plain}
\newtheorem{theorem}{Theorem}
\newtheorem{lemma}[theorem]{Lemma}
\newtheorem{proposition}[theorem]{Proposition}

\theoremstyle{definition}
\theorembodyfont{\normalfont}
\newtheorem{remark}{Remark}

\title{The Arity Hierarchy in the Polyadic $\mu$-Calculus{\footnote{The European Research Council has provided financial support under the European Community's Seventh Framework Programme (FP7/2007-2013) / ERC grant agreement no 259267.}}}
\author{Martin Lange  
\institute{School of Electrical Engineering and Computer Science, University of Kassel, Germany}
}
\def\titlerunning{The Arity Hierarchy in the Polyadic $\mu$-Calculus}
\def\authorrunning{M.~Lange}
\begin{document}
\maketitle

\begin{abstract}
The polyadic $\mu$-calculus is a modal fixpoint logic whose formulas define relations
of nodes rather than just sets in labelled transition systems. It can express
exactly the polynomial-time computable and bisimulation-invariant queries on
finite graphs. In this paper we show a hierarchy result with respect to expressive 
power inside the polyadic $\mu$-calculus: for every level of fixpoint alternation, greater
arity of relations gives rise to higher expressive power. The proof uses a diagonalisation 
argument. 
\end{abstract}

\addtocounter{footnote}{-1}

\section{Introduction}

The modal $\mu$-calculus $\mucalc{}$ is a well-studied logic \cite{ICALP::Kozen1982,BraStir01,BraStir04}, obtained by adding restricted second-order
quantification in the form of least and greatest fixpoints to a multi-modal logic interpreted over labelled transition systems. A formula of the
modal $\mu$-calculus is thus interpreted in a state of such transition systems which means that such formulas \emph{define} sets of states in 
transition systems. For example, $\nu X.\mu Y.\Mudiam{a}{}X \vee \Mudiam{b}{}Y$ defines the set of all states from which there is a path with labels
`$a$' and `$b$' that contains infinitely many occurrences of the symbol `$a$'. 

The polyadic $\mu$-calculus \mucalc{\omega} is a much less known extension of the modal $\mu$-calculus whose formulas define \emph{relations} rather than
sets of states. They are interpreted in a tuple of states rather than a single state, and there are modal operators for each position in this
tuple. Thus, one states ``the third state has an `$a$'-successors'' for instance rather than just ``there is an `$a$'-successors.'' Combining such simple
modal statements with fixpoint quantifiers yields an expressive logic with interesting applications: the polyadic $\mu$-calculus was first defined by
Andersen \cite{AndersenPMC:1994} and used as a logic for defining process equivalences like bisimilarity \cite{LL-FICS12,LLVG:TCS:2014}. Later it was 
re-invented by Otto under the name
\emph{Higher-Dimensional $\mu$-Calculus} \cite{Otto/99b} and shown to capture the complexity class P over bisimulation-invariant class of finite
graphs. I.e.\ a bisimulation-invariant property of finite graphs can be computed in polynomial time iff it is definable in \mucalc{\omega}.

There is a natural hierarchy in \mucalc{\omega} given by fragments of bounded arity. The polyadic $\mu$-calculus itself can be seen as a fragment of
FO+LFP, i.e.\ First-Order Logic extended with fixpoint quantifiers. The translation naturally extends the standard translation of modal logic into
first-order formulas with one free variable, seen as the point of reference for the interpretation of the property expressed by the modal formula.
Polyadic formulas get interpreted in tuples of states, hence they can be seen as special first-order formulas with several free variables. The 
arity of a polyadic formula is then the minimal number of free variables needed to express this property in FO+LFP or, equivalently, the length
of the tuples used to interpret the formula. 

The aim of this article is to show that the hierarchy formed by fragments of bounded arity, denoted \mucalc{1}, \mucalc{2}, \ldots is strict. This is 
not too surprising when taken literally: clearly, \emph{any} satisfiable but non-valid formula in \mucalc{k{+}1} is not equivalent to any formula
in \mucalc{k} since the former get interpreted in $k{+}1$-tuples and the latter only in $k$-tuples. We therefore need to employ a convention that
allows different fragments to be compared with respect to expressive power and still yields a meaningful hierarchy result. We consider formulas
that are interpreted in a single state at the top-level, regardless of their arity. I.e.\ we show that for every $k\ge1$ there is a 
\mucalc{k{+}1}-formula $\Phi_{k{+}1}$ such that there is no \mucalc{k}-formula $\psi$ which yields
\begin{displaymath}
\mathcal{T},(\underbrace{s,\ldots,s}_{k \text{\scriptsize\ times}}) \models \psi \quad \text{iff} \quad
\mathcal{T},(\underbrace{s,\ldots,s}_{k+1 \text{\scriptsize\ times}}) \models \Phi_{k+1}
\end{displaymath}
for all labelled transition systems $\mathcal{T}$ and all their states $s$.
 
Arity hierarchies have been studied before, most notably by Grohe for fixpoint extensions of first-order logic including FO+LFP \cite{APAL::Grohe1996}.
Even though each \mucalc{k} can be embedded into FO+LFP, the arity hierarchy in \mucalc{\omega} does not follow immediately from the one in
FO+LFP. Grohe constructs formulas of arity $k+1$ in FO+LFP -- they belong to the smaller FO+TC already -- and shows that they are not equivalent to any
formulas of arity $k$ in FO+LFP -- not even the much larger FO+sPFP. However, these witnessing formulas are not bisimulation-invariant since they
express a relation formed by the transitive closure of a clique relation and being a clique is clearly not bisimulation-invariant. Hence, these
witnessing formulas are not expressible in \mucalc{\omega} and therefore the arity hierarchy is not transferred immediately.

It could of course be checked whether the proof used to show the arity hierarchy in FO+LFP could be adapted to work for \mucalc{\omega} as well. It
would require the search for a similar witnessing property and the adaption of the Ehrenfeucht-Fra\"{\i}ss\'e argument to the polyadic $\mu$-calculus.
Such model comparison games exist for the modal $\mu$-calculus \cite{1996:tacas:stirling} but using them to obtain inexpressibility results has proved
to be quite difficult.

Instead we use a simple diagonalisation argument in order to obtain a strictness result regarding arity hierarchies. A $k$-ary formula $\varphi$ 
can be seen syntactically 
as a labelled transition system $\mathcal{T}_\varphi$, roughly based on the syntax-tree representation. We can then define a $k{+}1$-ary formula that 
simulates the evaluation of $\varphi$ on $\mathcal{T}_\varphi$ and accepts those $\mathcal{T}_\varphi$ which are not accepted by $\varphi$ itself. Hence,
we need to find a generic way of dualising the operators in $\varphi$. This is no particular problem, for instance, when one sees a disjunction then
one needs to check \emph{both} disjuncts, for a conjunction one only needs to check one of them. However, fixpoint formulas may hold or not because
of infinite recursive unfoldings through fixpoint operators. This needs to be dualised as well, and the only way that we can see to do this is to
equip the simulating formula with a fixpoint structure that is at least as rich as the one of the simulated formula. Consequently, we obtain an
arity hierarchy relative to the alternation hierarchy. This does not happen for extensions of First-Order Logic since it is known that there is
no alternation hierarchy: every FO+LFP formula can be expressed with a single least fixpoint operator only \cite{Imm:relqcp,STOC::Vardi1982}. The 
situation for modal logics is different: more fixpoint alternation generally gives higher expressive power, at least so in the modal $\mu$-calculus
\cite{CONCUR::Bradfield1996}, and presumably then so in \mucalc{\omega} as well.

The rest of this paper is organised as follows. In Section~\ref{sec:mucalc} we recall the polyadic $\mu$-calculus and necessary tools like fixpoint
alternation and model checking games. In Section~\ref{sec:hierarchy} we prove the hierarchy results, and in Section~\ref{sec:concl} we conclude with
a discussion on further work.


\section{The Polyadic $\mu$-Calculus}
\label{sec:mucalc}

\paragraph*{Labelled Transition Systems.}
Let $\Prop = \{p,q,\ldots\}$ and $\Act = \{a,b,\ldots\}$ be two fixed, countably infinite sets of atomic propositions and
action names. A labeled transition system (LTS) over $\Prop$ and $\Act$ is a tuple $\lts=(S,\Transition{}{}{},\lambda,s_I)$
where $S$ is a set of states, $\Transition{}{}{} \subseteq S \times \Act \times S$ is the transition relation, 
$\lambda: S \to 2^{\Prop}$ labels the states with atomic propositions, and $s_I$ is some designated starting state. 
We will write $\Transition{s}{a}{t}$ instead of $(s,a,t) \in \Transition{}{}{}$.

\paragraph*{The Syntax of \mucalc{\omega}.}
Let $\Var = \{X,Y,\ldots\}$ be an infinite set of second-order variables. The syntax of the polyadic modal $\mu$-calculus 
$\mucalc{\omega}$ is similar to that of the ordinary modal $\mu$-calculus. However, modalities and propositions
are relativised to a natural number pointing at a position in a tuple of states used to interpret the formula.

A \emph{replacement} is a $\kappa: \Nat \to \Nat$ which acts like the identity function on almost all arguments. We write
$\Nat \dashrightarrow \Nat$ to denote the space of all replacements. Such 
a replacement is then written as $\repl{\kappa(i_1) \leftarrow i_1,\ldots,\kappa(i_m) \leftarrow i_m}$ when $i_1 < \ldots < i_m$
are all those indices for which we have $\kappa(i_j) \ne i_j$. We will sometimes allow ourselves to deviate from this and to use
some shorter but equally intuitive notation for such functions. For instance $\repl{1 \leftrightarrow 2}$ should denote the swap between $1$ and
$2$, i.e.\ it abbreviates $\repl{2 \leftarrow 1, 1 \leftarrow 2}$.
 
For technical convenience, we define the logic directly in positive normal form. Formulas are then given by the grammar
\begin{displaymath}
\varphi \enspace ::= \enspace p(i) \mid \neg p(i) \mid X \mid \varphi \vee \varphi \mid \varphi \wedge \varphi \mid
\Mudiam{a}{i}\varphi \mid \Mubox{a}{i}\varphi \mid \mu X.\varphi \mid \nu X.\varphi \mid \kappa\varphi
\end{displaymath}
where $p \in \Prop$, $a \in \Act$, $1 \ge i \in \Nat$ and $\kappa$ is a replacement. We require that every second-order variable 
gets bound by a unique fixpoint quantifier $\mu$ or $\nu$. Then for every formula $\varphi$ there is a function 
$\mathit{fp}_\varphi$ which maps each second-order variable $X$ occurring in $\varphi$ to its unique binding formula 
$\mathit{fp}_\varphi(X)=\eta X.\psi$. 

The set $\sub{\varphi}$ of subformulas of $\varphi$ is defined as usual, with $\sub{\mu X.\varphi} = \{\mu X.\varphi\} \cup \sub{\varphi}$
for instance.

Later we will use the abbreviation $\ell \to \varphi$ when $\ell$ is a literal $q(i)$ or $\neg q(i)$. This behaves like ordinary implication
-- note that we have defined the logic in positive normal form and can therefore not simply introduce implication via negation -- for
such formulas when seen as $\bar{\ell} \vee \varphi$ where $\bar{ell}$ is the usual complementary literal to $\ell$.

The \emph{arity} of a formula $\varphi$, denoted $\ar{\varphi}$ is the largest index $i$ occurring in the operators $p(i)$, $\Mudiam{a}{i}$,
$\Mubox{a}{i}$ and $\repl{\kappa}$ in any of its subformulas. The fragment of arity $k$ is 
$\mucalc{k} := \{ \varphi \mid \ar{\varphi} \le k \}$. Hence, $\varphi := \nu X.\Mudiam{a}{1}\repl{2 \leftrightarrow 1}X$ has arity $2$
and it therefore belongs to all fragments $\mucalc{2}$, $\mucalc{3}$, etc., because it defines a relation of arity $2$ which can also
be seen as a relation of higher arity in which the 3rd, 4th, etc.\ components of its tuples are simply unrestrained.

\paragraph*{The Semantics of \mucalc{\omega}.}
Formulas of $\mucalc{k}$ are interpreted in $k$-tuples of states of a transition system $\lts = (S,\Transition{}{}{},\lambda,s_I)$.
An interpretation $\rho: \Var \to 2^{S^k}$ is neede in order to define this inductively and give a meaning to formulas with free variables.
For each $\mucalc{k}$-formula $\varphi$, $\sem{\varphi}{\rho}{\lts}$ is a $k$-ary relation of states in $\lts$, namely the
relation defined by $\varphi$ under the assumption that its free variables are interpreted by $\rho$.
\begin{align*}
\sem{p(i)}{\rho}{\lts} \enspace &:= \enspace \{ (s_1,\ldots,s_k) \mid p \in \lambda(s_i) \} \\ 
\sem{\neg p(i)}{\rho}{\lts} \enspace &:= \enspace \{ (s_1,\ldots,s_k) \mid p \not\in \lambda(s_i) \} \\ 
\sem{X}{\rho}{\lts} \enspace &:= \enspace \rho(X) \\ 
\sem{\varphi \vee \psi}{\rho}{\lts} \enspace &:= \enspace \sem{\varphi}{\rho}{\lts} \cup \sem{\psi}{\rho}{\lts} \\ 
\sem{\varphi \wedge \psi}{\rho}{\lts} \enspace &:= \enspace \sem{\varphi}{\rho}{\lts} \cap \sem{\psi}{\rho}{\lts} \\ 
\sem{\Mudiam{a}{i}\varphi}{\rho}{\lts} \enspace &:= \enspace \{ (s_1,\ldots,s_k) \mid 
  \exists t \text{ s.t. } \Transition{s_i}{a}{t} \text{ and } (s_1,\ldots,s_{i-1},t,s_{i+1},\ldots,s_k) \in \sem{\varphi}{\rho}{\lts} \} \\ 
\sem{\Mubox{a}{i}\varphi}{\rho}{\lts} \enspace &:= \enspace \{ (s_1,\ldots,s_k) \mid 
  \forall t: \text{ if } \Transition{s_i}{a}{t} \text{ then } (s_1,\ldots,s_{i-1},t,s_{i+1},\ldots,s_k) \in \sem{\varphi}{\rho}{\lts} \} \\ 
\sem{\mu X.\varphi}{\rho}{\lts} \enspace &:= \enspace \bigcap \{ R \subseteq S^k \mid \sem{\varphi}{\rho[X \mapsto R]}{\lts} \subseteq R \} \\ 
\sem{\nu X.\varphi}{\rho}{\lts} \enspace &:= \enspace \bigcup \{ R \subseteq S^k \mid \sem{\varphi}{\rho[X \mapsto R]}{\lts} \supseteq R \} \\
\sem{\kappa\varphi}{\rho}{\lts} \enspace &:= \enspace \{ (s_{\kappa(1)},\ldots,s_{\kappa(k)}) \mid (s_1,\ldots,s_k) \in \sem{\varphi}{\rho}{\lts} \}
\end{align*}
Note that the partial order $\subseteq$ makes $S^k$ a complete lattice with meets and joins given by $\bigcap$ and $\bigcup$, and the 
semantics of fixpoint formulas is then well-defined according to the Knaster-Tarski Theorem \cite{Kna28,Tars55}. 

We also write $\lts, s_1,\ldots,s_k \models_\rho \varphi$ instead of $(s_1,\ldots,s_k) \in \sem{\varphi}{\rho}{\lts}$. If $\varphi$ has
no free second-order variables then we also drop $\rho$. In Section~\ref{sec:hierarchy} we will often consider situations with tuples of
the form $(s,\ldots,s)$ of some length $k$ derivable from the context. We will then simply write $\lts, s \models \varphi$ as a short form
for $\lts, s, \ldots, s \models \varphi$.

Two formulas $\varphi,\psi \in \mucalc{k}$ are \emph{equivalent}, written $\varphi \equiv \psi$, if $\sem{\varphi}{\rho}{\lts} = \sem{\psi}{\rho}{\lts}$
for any $\lts$ and corresponding variable interpretation $\rho$. Note that two formulas can be equivalent even if they do not belong to the same
arity fragment: if $\varphi \in \mucalc{k}$ and $\psi \in \mucalc{k'}$ and $k \ne k'$ then $\varphi, \psi \in \mucalc{\max \{k,k'\}}$, i.e.\ we
can interpret the one of smaller arity as a formula of larger arity that simply does not constrain the additional elements in the tuples of the
relation it defines.

\paragraph*{Examples.}
The standard example of a \mucalc{\omega}-formula, indeed a \mucalc{2}-formula, is the one defining \emph{bisimilarity}.
\begin{displaymath}
\varphi_\sim \enspace := \enspace \nu X.\big(\bigwedge\limits_{p \in \Prop} p(1) \to p(2)\big) \wedge 
(\bigwedge\limits_{a \in \Act} \Mubox{a}{1}\Mudiam{a}{2}X\big) \wedge \repl{1 \leftrightarrow 2}X
\end{displaymath} 
It is indeed the case that $\lts, s,t \models \varphi_\sim$ iff $s \sim t$, i.e.\ $s$ and $t$ are bisimilar in $\lts$.

As a second example consider an \lts with an edge relation $\mathsf{flight}$ and two atomic propositions $\mathsf{warm}$ and $\mathsf{safe}$. 
When seeing the nodes of the LTS as cities (which can or cannot be warm and/or safe and are potentially linked by direct flight connections),
then 
\begin{displaymath}
\repl{3 \leftarrow 1}\varphi_\sim \wedge \Mudiam{\mathsf{flight}}{2} \mu X. \mathsf{warm}(2) \wedge \mathsf{safe}(2) \wedge 
\Mudiam{\mathsf{flight}}{1}\varphi_\sim \wedge (\repl{3 \leftarrow 1}\varphi_\sim \vee \Mubox{\mathsf{flight}}{2}X) \wedge \repl{2 \leftarrow 3}X
\end{displaymath} 
yields all triples $(s,t,u)$ of cities such that there is a roundtrip from $t$ which only traverses through warm and safe cities that can be reached
from city $s$ in one step -- in case someone in $s$ wants to come and visit -- such that the trip can be traversed in both directions. This 
description of course uses equality (``roundtrip'') on cities
which is not available in the logic. Instead we use bisimilarity in the formula, so for instance ``roundtrip from $t$'' is to be understood as
a trip starting in $t$ and ending in a city that cannot be distinguished from $t$ with the means of bisimilarity.

\paragraph*{Fixpoint Alternation.} The proof of the arity hierarchy carried out in Section~\ref{sec:hierarchy} needs a closer look at the dependencies
of fixpoints inside a formula. This phenomenon is well-understood leading to the notion of alternation hierarchy \cite{EL86,LICS::Niwinski1988}. 
We give a brief intoduction to fixpoint alternation that is sufficient for the purposes of the next section.

Let $k \ge 1$ and $\varphi \in \mucalc{k}$ be fixed. For two variables $X,Y \in \sub{\varphi}$ we write $X \ge_\varphi Y$ if $X$ has a free occurrence
in $\fp{\varphi}(Y)$. We use $>_\varphi$ to denote the strict part of its transitive closure. E.g.\ in 
\begin{displaymath}
\varphi := \mu X.p(2) \vee \Mudiam{b}{1}(\nu Y.q(1) \wedge \nu Y'.(\mu Z.Y' \vee \Mudiam{a}{1}Z) \wedge \Mubox{b}{2}Y)
\end{displaymath}
we have $X >_\varphi Y >_\varphi Y' >_\varphi Z$ even though there is no free occurrence of $X$ in the fixpoint formula for $Z$. 

Names of variables do not matter but their fixpoint types do. So we abstract this chain of fixpoint dependencies into a chain 
$\mu >_\varphi \nu >_\varphi \nu >_\varphi \mu$. The \emph{alternation type} of a formula is a maximal descending chain of variables 
(represented by their fixpoint types) such that adjacent types in this chain are different. The alternation type of $\varphi$ above is
therefore just $(\mu,\nu,\mu)$. We then define the
\emph{alternation hierarchy} as follows: $\Sigma^k_m$, respectively $\Pi_m^k$ consists of all formulas of arity $k$ and alternation type of 
length at most $m$ such that the $m$-th last in this chain is $\mu$, respectively $\nu$, if it exists. For instance, the formula $\varphi$
above belongs to $\Sigma^2_3$ and thefore also to $\Sigma^2_m$ and $\Pi_m^2$ for all $m > 3$. It does not belong to $\Pi^2_2$.

Each variable $X$ occuring in $\varphi$ is also given an \emph{alternation depth} $\ad{\varphi}{X}$. It is the index in a maximal chain of
dependencies $X_m >_\varphi \ldots >_\varphi X_1$ such that adjacent variables have different fixpoint types. E.g.\ in the example above we have
$\ad{\varphi}{X} = 3$, $\ad{\varphi}{Y} = \ad{\varphi}{Y'} = 2$ and $\ad{\varphi}{Z} = 1$.
 
The next observation is easy to see.

\begin{lemma}
\label{lem:ad}
Let $\varphi \in \Sigma_m^k$ and $X \in \sub{\varphi}$ be one of its fixpoint variables. Then the fixpoint type of $X$ is uniquely determined by $\ad{\varphi}{X}$, namely it is $\mu$ if $m$ and $i$ are both odd or both even, otherwise it is $\nu$.
\end{lemma}

\paragraph*{Model Checking Games.} We briefly recall model checking games for the polyadic $\mu$-calculus \cite{LL-FICS12}. They are defined in
the same style as the model checking games for the modal $\mu$-calculus \cite{Stirling95} as a game played between players \verifier and \refuter
on the product space of an LTS and a formula. Such games can be used to reason about the satisfaction of a formula in a structure since both
satisfaction and non-satisfaction are reduced to the existence of winning strategies for one of the players in these model checking games.

As with the modal $\mu$-calculus games, the model checking games for the polyadic $\mu$-calculus are nothing more than parity games. However,
they are played using $k$ pebbles in the LTS and one pebble on the set of subformulas of the input formula. Hence, a configuration is a 
$k+1$-tuple written $s_1,\ldots,s_k \vdash \psi$ where the $s_i$ are states of the underlying LTS $\lts = (S,\Transition{}{}{},\lambda,s_I)$
and $\psi$ is a subformula of the underlying formula $\varphi$.

The rules are as follows.
\begin{itemize}
\item In a configuration of the form $s_1,\ldots,s_k \vdash \psi_1 \vee \psi_2$, player \verifier chooses an $i \in \{1,2\}$ and the play
      continues with $s_1,\ldots,s_k \vdash \psi_i$. Intuitively, \verifier moves the formula pebble to a disjunct from the current disjunction.
\item Likewise, in a configuration of the form $s_1,\ldots,s_k \vdash \psi_1 \wedge \psi_2$, player \refuter chooses such an $i$. Here, this can
      be seen as refuter moving the formula pebble.
\item In a configuration of the form $s_1,\ldots,s_k \vdash \Mudiam{a}{i}\psi$, player \verifier chooses a $t$ such that $\Transition{s_i}{a}{t}$
      and the play continues with $s_1,\ldots,s_{i-1},t,s_{i+1},\ldots,s_k \vdash \psi$. Intuitively, \verifier moves the $i$-the state pebble along
      an outgoing $a$-transition. The other $k-1$ pebbles that are on states remain where they are. The formula pebble is also moved to the next
      subformula.
\item Likewise, in a configuration of the form $s_1,\ldots,s_k \vdash \Mubox{a}{i}\psi$, player \refuter chooses such a $t$.
\item In a configuration of the form $s_1,\ldots,s_k \vdash \eta X.\psi$ or $s_1,\ldots,s_k \vdash X$ such that $\mathit{fp}_\varphi(X) = \eta X.\psi$,
      the formula pebble is simply moved to $\psi$, i.e.\ the play continues with $s_1,\ldots,s_k \vdash \psi$.
\end{itemize}
A player wins a play if the opponent cannot carry out a move anymore. Moreover, \verifier wins a play that reaches a configuration of the
form $s_1,\ldots,s_k \vdash q(i)$ if $q \in \lambda(s_i)$. If, on the other hand, $q \not\in \lambda(s_i)$ then player \refuter wins this play.
Finally, there are infinite plays, and the winner is determines by the necessarily unique outermost fixpoint variable (i.e.\ the largest with 
respect to $>_\varphi$) that occurs infinitely often in this play. If its fixpoint type is $\nu$, then \verifier wins, otherwise it is $\mu$ and
\refuter wins.

The main advantage of these model checking games is the characterisation of the satisfaction relation via winning strategies in parity games
(which they essentially are).

\begin{proposition}[\cite{LL-FICS12}]
Player \verifier has a winning strategy in the game in $\lts$ and a closed $\varphi$, starting in the configuration $s_1,\ldots,s_k \vdash \varphi$
iff $\lts, s_1,\ldots,s_k \models \varphi$.
\end{proposition}


\section{The Arity Hierarchy}
\label{sec:hierarchy}

\subsection{The Principle Construction}
\label{sec:construction}
The aim of this section is to show that  $\mucalc{1},\mucalc{2},\ldots$ forms a strict hierarchy with respect to expressive power.
The principles underlying the proof are easily explained: first we associate with each \mucalc{k}-formula $\varphi$ an LTS $\lts_\varphi$
with a designated starting state which we also call $\varphi$. Then we construct a closed \mucalc{k+1}-formula that, when given a $\lts_\varphi$, 
reads off what $\varphi$ is from $\lts_\varphi$ and simulates its evaluation on it, checking that it does \emph{not} hold on itself.

We first present the constructions principally, then discuss what results are achieved with the details of these constructions, and
finally optimise the constructions such that the desired hierarchy result is achieved. We use a singleton $\Act$ wich means that we simply
write $\Transition{s}{}{t}$ instead of $\Transition{s}{a}{t}$ for the single action name `$a$'. Likewise, we write
$\Mudiamanon{i}$ and $\Muboxanon{i}$ instead of $\Mudiam{a}{i}$ and $\Mubox{a}{i}$.

\paragraph*{Construction of $\lts_\varphi$.} Let $k \ge 1$ be fixed and take an arbitrary closed $\varphi \in \mucalc{k}$. We assume that the
set of propositions underlying $\varphi$ is $\Prop = \{q_0,q_1,q_2,\ldots\}$. The construction of
$\lts_\varphi$ is largely based on the syntax-tree, respectively syntax-DAG of $\varphi$. We have 
$\lts_\varphi = (\sub{\varphi},\Transition{}{}{},\lambda,\varphi)$ with transitions given as follows.
\begin{align*}
\Transition{\psi_1 \varodot \psi_2}{}{}&{\psi_i} \quad && \text{for every } \psi_1 \varodot \psi_2 \in \sub{\varphi}, \varodot \in \{\wedge,\vee\} \text{ and every } i \in \{1,2\} \\
\Transition{\varodot\psi}{}{}&{\psi} \quad && \text{for every } \varodot\psi \in \sub{\varphi}, \varodot \in \{\Mudiamanon{i},\Muboxanon{i},\kappa\} \text{ and every } i \in \{1,\ldots,k\} \\
\Transition{\eta X.\psi}{}{}&{\psi} \quad && \text{for every } \eta X.\psi \in \sub{\varphi} \text{ and } \eta \in \{\mu,\nu\} \\
\Transition{X}{}{}&{\psi} \quad &&\text{for every } X \in \sub{\varphi} \text{ such that } \fp{\varphi}(X) = \eta X.\psi
\end{align*}
Thus, the graph structure of $\lts_\varphi$ is indeed almost the one of the syntax-DAG of $\varphi$ except for additional edges from fixpoint
variables to their defining fixpoint formula.

The labelling of the nodes in $\lts_\varphi$ remains to be defined. Remember that the ultimate goal is to construct a formula $\Phi^{k+1}$ which
simulates the evaluation of $\varphi$ on $\lts_\varphi$. We will use $k$ pebbles in order to simulate the $k$ pebbles used in $\varphi$, and
one additional pebble in order to store the subformula that is currently in question. Note that the satisfaction of a (closed) formula on
an LTS only depends on the satisfaction of its subformulas. The position of this additional pebble will determine which subformula is
currently evaluated. We therefore need to make the kind of subformula at a node in $\lts_\varphi$ visible to a formula that is interpreted over
it. This is what the state labels will be used for. Let 
\begin{displaymath}
\Prop_0 := \{ \propPLit{j}{i}, \propNLit{j}{i} \mid 1 \le i \le k, j \in \Nat \} \cup \{\propAnd, \propOr \} \cup 
\{ \propDiam{i}, \propBox{i} \mid 1 \le i \le k \} \cup \{ \propFP{i} \mid 0 \le i \le m \} \cup \{ \propRepl{\kappa} \mid \kappa \in \Nat \dashrightarrow \Nat \}\ .
\end{displaymath}
The labelling in $\lts_\varphi$ is given as follows. Note that $\Prop$ is countably infinite.
\begin{align*}
\propPLit{j}{i} &\in \lambda(q_j(i)) &&\text{for every positive literal } q_j(i) \in \sub{\varphi} \\ 
\propNLit{j}{i} &\in \lambda(q_j(i)) &&\text{for every negative literal } \neg q_j(i) \in \sub{\varphi} \\ 
\propAnd &\in \lambda(\psi_1 \wedge \psi_2) && \text{for every } \psi_1 \wedge \psi_2 \in \sub{\varphi} \\
\propOr &\in \lambda(\psi_1 \vee \psi_2) && \text{for every } \psi_1 \vee \psi_2 \in \sub{\varphi} \\
\propDiam{i} &\in \lambda(\Mudiamanon{i}\psi) && \text{for every } \Mudiamanon{i}\psi \in \sub{\varphi}, 1 \le i \le k \\
\propBox{i} &\in \lambda(\Muboxanon{i}\psi) && \text{for every } \Muboxanon{i}\psi \in \sub{\varphi}, 1 \le i \le k \\
\propRepl{\kappa} &\in \lambda(\kappa\psi) && \text{for every } \kappa\psi \in \sub{\varphi}, \kappa: \Nat \dashrightarrow \Nat \\
\propFP{i} &\in \lambda(\eta X.\psi),\lambda(X) && \text{for every } \eta X.\psi,X \in \sub{\varphi}, \eta \in \{\mu,\nu\} 
\text{ with } \ad{\varphi}{X} = i
\end{align*}
With those labels a formula can see what the subformula at a node is that it is interpreted over, for instance whether it is a formula with a
replacement modality as the principle operator, etc.

\paragraph{The construction of the simulating formulas.}

Next we construct formulas that simulate a $\varphi$ on its own LTs representation $\lts_\varphi$ and check that they do not satisfy themselves.
The trick is simple: if $\varphi \in \mucalc{k}$ then we use $k$ pebbles to simulate what $\varphi$ would do with its $k$ pebbles, and one additional
pebble to check wich subformula we are currently evaluating. We let this pebble move through the syntax-DAG in a form that is dual to the semantics
of the actual operators in the underlying $\varphi$; for instance in a conjunction we look for one conjunct, in a disjunction we continue with both disjuncts. We will use several fixpoint variables to dualise the fixpoint condition similar to the way it is done in the Walukiewcz formulas that
express the winning conditions in parity games \cite{journals/tcs/Walukiewicz02}.

Let $m \ge 0$ and $k \ge 1$ be fixed. We construct a formula $\Phi^{k+1}_m \in \mucalc{k+1}$ as follows.
\begin{align*}
\Phi^{k+1}_m \enspace := \enspace \nu X_m. \mu X_{m-1} \ldots \eta X_1.\Big(
&\bigwedge\limits_{i=1}^k\bigwedge\limits_{j \in \Nat} \propPLit{j}{i}(k+1) \to \neg q_j(i) \\
\wedge \enspace &\bigwedge\limits_{i=1}^k\bigwedge\limits_{j \in \Nat} \propNLit{j}{i}(k+1) \to q_j(i) \\
\wedge \enspace &\propAnd(k+1) \to \Mudiamanon{k+1}X_1 \\
\wedge \enspace &\propOr(k+1) \to \Muboxanon{k+1}X_1 \\
\wedge \enspace &\bigwedge\limits_{i=1}^k \propDiam{i}(k+1) \to \Muboxanon{i}\Muboxanon{k+1}X_1 \\
\wedge \enspace &\bigwedge\limits_{i=1}^k \propBox{i}(k+1) \to \Mudiamanon{i}\Muboxanon{k+1}X_1 \\
\wedge \enspace &\bigwedge\limits_{\kappa \in \Nat \dasharrow \Nat} \propRepl{\kappa}(k+1) \to \kappa\Muboxanon{k+1}X_1 \\
\wedge \enspace &\bigwedge\limits_{i=1}^m \propFP{i}(k+1) \to \Muboxanon{k+1}X_i \Big)
\end{align*}
where $\eta = \nu$ if $m$ is odd and $\eta = \mu$ otherwise. 

\begin{remark} \normalfont
\label{rem:noformula}
Of course, $\Phi^{k+1}_m$ is not a formula strictly speaking because of the potentially
infinite conjunctions in the first two clauses. There is an easy way to fix this: we assume a finite set of atomic propositions $\{p,q,\ldots\}$. 
Then a finite conjunction obviously suffices and $\Phi^{k+1}_m$ is indeed a formula. However, we need to address the issue of choice of atomic
propositions in Section~\ref{sec:fixsig} below anyway. So for the moment we simply accept the small flaw about infinite conjunctions as an
intermediate step and as a means to separate the principles from the details in this construction.

Note that this problem does not arise in the clause with the $\propRepl{\kappa}$ since $k$ is fixed, and $\kappa$ can at most change the first $k$ pebbles.
Hence, there are only finitely many such $\kappa$.
\end{remark}

We need two observations about $\Phi^{k+1}_m$. The first, a syntactic one, is easy to verify.

\begin{lemma}
\label{lem:inPi}
For every $m \ge 0$ and every $k \ge 1$ we have $\Phi_m^{k+1} \in \Pi_{m}^{k+1}$.
\end{lemma}

The second one is of a semantic nature and states that $\Phi^{k+1}_m$ does what it is supposed to do.

\begin{lemma}
\label{lem:correct}
Let $m \ge 0$, $k \ge 1$ and $\varphi \in \Sigma_m^k$. Then we have $\lts_\varphi, \varphi \models \Phi_m^{k+1}$ iff $\lts_\varphi, \varphi \not\models \varphi$.
\end{lemma}

\begin{proof}
We argue using model checking games for \mucalc{\omega}. 

``$\Leftarrow$'' Suppose we have $\lts_\varphi, \varphi \not\models \varphi$, i.e.\ \refuter has a winning strategy for the game $\mathcal{G}$ 
played on $\lts_\varphi$, $k$ pebbles initially placed on the node $\varphi$ in it, and the \mucalc{k}-formula $\varphi$ itself. This gives rise to a
strategy for player \verifier in the game $\mathcal{G}'$ played on $\lts_\varphi$, now $k+1$ pebbles placed on node $\varphi$ initially, and the formula
$\Phi^{k+1}_m$. The fact that each node in $\lts_\varphi$ satisfies exactly on atomic proposition of the kind $p^*$ and at most one $i$ or
at most on $\kappa$ means that any play which \refuter does not lose immediately selects a clause in $\Phi^{k+1}_m$, carries out some operation
on the pebbles and then loops through some fixpoint variable. 

It is not hard to see that \verifier can use \refuter's strategy from $\mathcal{G}$ to follow the operations carried out on the pebbles prescribed
by each clause without losing. For instance, if the third clause demands her to choose a successor for the $k+1$-st pebble then she takes the one
that represents the conjunct that \refuter would chose in the same situation in $\mathcal{G}$. This way, every play in $\mathcal{G}'$ that conforms
to her strategy has an underlying play in $\mathcal{G}$ that conforms to \refuter's strategy there. If that one is won by \refuter because \verifier
got stuck at some point then this can only be because the play reached a position of the form $(s_1,\ldots,s_k) \vdash \Mudiamanon{i}\psi$ and
$s_i$ has no successor. In the corresponding play in $\mathcal{G}'$, pebble $k+1$ will be on a node with label $\propDiam{i}$, and this requires 
\refuter to move the $i$-th pebble to a successor which equally he cannot. Notice that the clause with $\propDiam{i}$ contains the operator 
$\Muboxanon{}$ and vice-versa. Thus, \verifier wins the corresponding play in $\mathcal{G}'$.

Suppose that the underlying play in $\mathcal{G}$ is won by \refuter because the largest fixpoint variable $X$ that is seen infinitely often is of
type $\mu$. Then we must have $\ad{\varphi}{X} = i$ for some $i$, and then the play in $\mathcal{G}'$ will infinitely often go through positions that 
are labelled with $\propFP{i}$, and it will eventually not go through positions that are labelled with $\propFP{i'}$ with $i' > i$ anymore.
All that remains to be seen in this case is that the largest variable seen infinitely often in the play on $\Phi^{k+1}_m$ is of type $\nu$. This is
a direct consequence of Lemma~\ref{lem:ad}. Hence, \verifier wins such plays, too, which shows that her strategy derived from \refuter's winning
strategy in $\mathcal{G}$ is winning for her in $\mathcal{G}'$.

``$\Rightarrow$'' This is shown by contraposition in the same way now assuming a winning strategy for \verifier in the game on $\lts_\varphi$ and
$\varphi$ and turning it into a winning strategy for \refuter in the game on $\lts_\varphi$ and $\Phi^{k+1}_m$.
\end{proof}

\begin{lemma}
Let $m \ge 0$ and $k \ge 1$. There is no $\varphi \in \Sigma_m^k$ such that $\varphi \equiv \Phi_m^{k+1}$.
\end{lemma}

\begin{proof}
Suppose there was such a $\varphi$. Then we would have
\begin{displaymath}
\lts_\varphi, \varphi \models \varphi \quad \text{iff} \quad \lts_\varphi, \varphi \models \Phi_m^{k+1} \quad \text{iff} \quad 
\lts_\varphi, \varphi \not\models \varphi
\end{displaymath}
first because of the assumed equivalence and second because of Lemma~\ref{lem:correct}.
\end{proof}

Thus, we could summarise the findings from these lemmas and also uses the observation that the entire construction is equally possible
for formula in $\Pi_m^{k}$ then yielding a $\Phi_m^{k+1} \in \Sigma_m^{k+1}$. Then we get that
for all $m \ge 0$ and $k \ge 1$ we have $\Sigma_m^{k} \not\supseteq \Pi_m^{k+1}$ and $\Pi_m^k \not\supseteq \Sigma_m^{k+1}$. Consequently, we have
$\Sigma^k_m \subsetneq \Sigma_{m+1}^{k+1}$ and $\Pi_m^k \subsetneq \Pi_{m+1}^{k+1}$.

The reason why we do not formally state this as a theorem (yet) is discussed next.

\subsection{The Hierarchy over a Fixed Small Signature}
\label{sec:fixsig}

Consider what is happing with the set of atomic propositions in the construction of the previous Section~\ref{sec:construction}. We have already
seen in Remark~\ref{rem:noformula} that the construction does not work for an infinite set of atomic propositions $\Prop$. Even if this is finite,
then the construction does work but it has the following effect: we simulate a formula with $k$ pebbles over $\Prop$ by a formula with $k+1$
pebbles over $\Prop \cup \Prop_0$. It is not surprising that we obtain formulas over this extended signature which cannot be expressed over the
smaller one. In order to argue that the hierarchy of inexpressibility as laid out in the previous section is truly meaningful
we would need $\Prop = \Prop \cup \Prop_0$ or, at least, that the two sets have equal cardinality so that some bijection between them could be 
used as an encoding. 

In the following we will show how the construction can be fixed such that it works over a fixed finite set 
\begin{displaymath}
\Prop_1 := \{ \propPLitNo, \propNLitNo, \propAnd, \propOr, \propDiamNo, \propBoxNo, \propFPNo, \propReplNo, \propSwapNo, \propStop \}
\end{displaymath}
of atomic propositions. Thus, we do not encore the index of propositions, the level in the fixpoint hierarchy, and the kind of operation on pebbles
in those propositions anymore. Instead we will encode this missing information in the graph structure of $\lts_\varphi$ (rather than in its labels). 
For the replacement modalities $\kappa$ we need a little preparation.

A replacement $\kappa$ is called simple if it is of the form $\repl{i \leftarrow j}$ or $\repl{i \leftrightarrow j}$. A formula is called
\emph{normalised} if every replacement in it is simple. The following is a simple consequence of the fact that every function 
$\kappa: \Nat \dashrightarrow \Nat$ that leaves all numbers greater than $k$ untouched, can be constructed by a sequence of swaps between
$i,j \le k$, followed by some simple mappings from some $i$ to a $j$.

\begin{lemma}
Let $m \ge 0$, $k \ge 1$. Every $\varphi \in \Sigma_m^k$, respectively $\Pi_m^k$ is equivalent to a normalised $\varphi' \in \Sigma_m^k$, 
respectively $\Pi_m^k$.
\end{lemma}
 
We therefore assume that from now on, all formulas $\varphi$ to be simulated are normalised. We change the construction of $\lts_\varphi$ as
follows. 
\begin{enumerate}
\item Suppose there is a state $s$ of the form $q_j(i)$ or $\neg q_j(i)$, necessarily labeled with $\propPLit{j}{i}$ or $\propNLit{j}{i}$. Replace
      the proposition by $\propPLitNo$, respectively $\propNLitNo$, and add a new finite path of length $i+j$ to this node such that the $(i-1)$-st
      new state has the label $\propStop$.
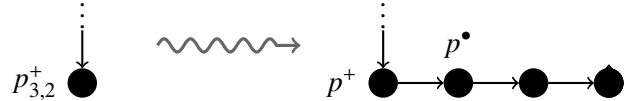
      \begin{center}
       \begin{tikzpicture}[thick,->, every state/.style={fill=black,minimum size=3mm,draw}]

          \node[state,label=180:{$\propPLit{3}{2}$}] (oldv) {};

          \path[<-,draw,shorten >=3mm] (oldv) -- ++(0,1) node {$\vdots$};

          \node[state,label=180:{$\propPLitNo$}] (newv) [right of=oldv,node distance=4cm] {};
          \path[<-,draw,shorten >=3mm] (newv) -- ++(0,1) node {$\vdots$};
          
          \node[state,label=90:{$\propStop$}] (1) [right of=newv] {};
          \node[state]                        (2) [right of=1]    {};
          \node[state]                        (3) [right of=2]    {};
          \node[state]                        (4) [right of=2]    {};
          \node[state]                        (5) [right of=2]    {};

          \path[->] (newv) edge (1) (1) edge (2) (2) edge (3) (3) edge (4) (4) edge (5); 

          \path[->,very thick, draw=black!60,decorate, decoration={coil, aspect=0}] ($(oldv)+(1,.5)$) -- ($(newv)+(-1.1,.5)$);
 
       \end{tikzpicture}
      \end{center}
      Let $\kappa^{\leftarrow}_{\mathsf{cyc}} := \repl{k \leftarrow 1,1 \leftarrow 2,\ldots, k-1 \leftarrow k}$ and consider the formula 
      \begin{equation}
        \label{eq:searchi}
         \mathit{searchPeb} := (\propStop \wedge q_j(1)) \vee \kappa^{\leftarrow}_{\mathsf{cyc}}\Mudiamanon{k+1}((\propStop \wedge q_j(1)) \vee \kappa^{\leftarrow}_{\mathsf{cyc}}\Mudiamanon{k+1}((\propStop \wedge q_j(1)) \vee \ldots \kappa^{\leftarrow}_{\mathsf{cyc}}\Mudiamanon{k+1}(\propStop \wedge q_j(1))\ldots)) 
      \end{equation}
      with $k-1$ occurrences of $\kappa^{\leftarrow}_{\mathsf{cyc}}$.
      It is true in $s$ at pebble $k+1$ with the additional path iff $q_j(i)$ was true in $s$ with the original construction. The first part of this
      new path is used to shift the pebbles until the $i$-th has become the first and then, instead of checking whether the $i$-th pebble is on
      a state satisfiying $q_j$, we can now check the first one instead. Note that this formulas moves pebble number $k+1$ along this new path but the 
      other $k$ pebbles remain where they are; apart from being cyclically changed around.

      We can of course equally construct such a formula that mimicks the checking of $\neg q_j(i)$.

      Finally, we also need to use the remaining path to read off the encoding of $j$. This can easily be done as follows.
      \begin{displaymath}
        \mathit{searchProp}_j := q_0(1) \vee \Mudiamanon{k+1}(q_1(1) \vee \Mudiamanon{k+1}(q_2(1) \vee \ldots \Mudiamanon{k+1}(q_{h-2}(1) \vee \Mudiamanon{k+1}q_{h-1}(1))\ldots))  
      \end{displaymath}
      This formula is then used instead of $q_j(1)$ in (\ref{eq:searchi}), and the resulting formula is used instead of $q_j(i)$ in
      the clause for $\propNLit{j}{i}$ in $\Phi^{k+1}_m$. Hence, this clase simply becomes
      \begin{displaymath}
       \ldots \wedge \propNLit{j}{i} \to \mathit{searchPeb}[\mathit{searchProp}/q_j(1)]
      \end{displaymath}
      where $\psi[\chi/\chi']$ denotes the formula that results from $\psi$ by replacing every subformula $\chi'$ with $\chi$.

\item An edge of the form $\Transition{\eta X.\psi}{}{\psi}$ or $\Transition{X}{}{\psi}$ is replaced in similar style by a sequence of $i$ edges,
      marking the last state after them with $\propStop$. Then we can replace the label $\propFP{i}$ with $\propFPNo$ in the first state, and the
      corresponding clause in $\Phi^{k+1}_m$ with 
      \begin{displaymath}
       \ldots \wedge \propFPNo \to \Muboxanon{k+1}((\propStop \wedge X_1) \vee \Muboxanon{k+1}((\propStop \wedge X_2) \vee \Muboxanon{k+1}(\ldots \vee \Muboxanon{k+1}(\propStop \wedge X_m)\ldots)))
      \end{displaymath}

\item An edge of the form $\Transition{\Mudiamanon{i}\psi}{}{\psi}$ or $\Transition{\Muboxanon{i}\psi}{}{\psi}$ is replaced by a sequence of 
      $2i$ edges via new states, and $\propStop$ must hold after $i$ and after $2i$ steps. The trick to use here is to cycle the first $k$ pebbles
      until the $i$-th one becomes the first, then execute the corresponding action for the $i$-pebble on the first one instead, and then cycle them
      back again. Let $\kappa^{\leftarrow}_{\mathsf{cyc}}$ be as above and 
      $\kappa^{\rightarrow}_{\mathsf{cyc}} := \repl{2 \leftarrow 1,3 \leftarrow 2,\ldots, 1 \leftarrow k}$. Then we can replace the clause for
      $\propDiam{i}$ in $\Phi^{k+1}_m$ by
      \begin{displaymath}
        \ldots \wedge \propDiamNo \to (\propStop \wedge \Muboxanon{k+1}(\Muboxanon{1}\mathit{goBack} \vee \kappa^{\leftarrow}_{\mathsf{cyc}}(\propStop \wedge 
         \Muboxanon{k+1}(\Muboxanon{1}\mathit{goBack} \vee \ldots (\propStop \wedge \Muboxanon{k+1}\Muboxanon{1}\mathit{goBack})\ldots)))
      \end{displaymath}
      with exactly $k-1$ occurrences of $\kappa^{\leftarrow}_{\mathsf{cyc}}$ and
      \begin{displaymath}
        \mathit{goBack} := \Muboxanon{k+1}((\propStop \wedge X_1) \vee \kappa_{\mathsf{cyc}}^{\rightarrow}\Muboxanon{k+1}((\propStop \wedge X_1) \vee
         \ldots \Muboxanon{k+1}(\propStop \wedge X_1)\ldots))
      \end{displaymath}
      with exactly $k-1$ occurrences of $\kappa_{\mathsf{cyc}}^{\rightarrow}$.

      Likewise, we can use the same trick to eliminate the dependence on $i$ of the formula $\Phi^{k+1}_m$ in the clause for $\propBox{i}$ which is
      equally replaced by $\propBoxNo$, and those paths of length $2i$ can be used to decode the value $i$ from the graph structure instead of 
      reading it straight off the atomic proposition.

\item Finally, we can use the same trick in a slightly more elaborate fashion to handle replacement modailities of the form $\repl{i \leftarrow j}$
      and $\repl{i \leftrightarrow j}$. We mark nodes in $\lts_\varphi$ that correspond to the form by $\propReplNo$ and those that correspond to 
      the latter by $\propSwapNo$. A swap of the form $\repl{i \leftrightarrow j}$ can be handled as follows: assume $i < j$. 
      \begin{enumerate}
       \item Cyclically shift the pebbles $1,\ldots,k$ for $i$ positions to the left.
       \item Cyclically shift the pebbles $2,\ldots,k$ for $j-i-1$ positions to the left.
       \item Swap pebbles 1 and 2.
       \item Cyclically shift the pebbles $2,\ldots,k$ for $j-i-1$ positions to the right.
       \item Cyclically shift the pebbles $1,\ldots,k$ for $i$ positions to the right.
      \end{enumerate}
      Hence, we replace a transition of the form $\Transition{\repl{i \leftarrow j}\psi}{}{\psi}$ by a path of length $2j-2$ and mark the states
      at positions $i$, $j-1$, $2j-2$ and the last one with $\propStop$ so that we can, like above, construct a formula that mimicks the five steps
      above to carry out the swapping of pebbles $i$ and $j$.

      The construction for replacements of the form $\repl{i \leftarrow j}$ is similar. Again, the trick is to cycle $i$ and $j$ to positions 
      1 and 2, carry the replacement out on these fixed positions, and cycle the pebbles back again. These eliminates the dependence of $\Phi^{k+1}_m$
      on propositions which carry such a value.
\end{enumerate}
With this being done, $\Phi^{k+1}_m$ becomes a formula that is defined over a fixed set $\Prop_1$ of atomic propositions of size 10, and we can
use it to simulate formulas $\varphi \in \Sigma_m^k$ over the same $\Prop_1$. Then the inexpressibility result of the previous section becomes 
meaningful. Using standard encoding techniques we can break the resul down to \mucalc{\omega} over two atomic propositions only, using
binary encoding, or a single one, using unary encoding. The atomic propositions can also be eliminated entirely by appending certain finite
trees to the states in which they hold such that these trees are checkable using fixpoint-free formulas of \mucalc{\omega}. Hence, we get the
following.

\begin{theorem}
\label{thm:hierarchy}
For all $m \ge 0$ and $k \ge 1$ we have $\Sigma_m^{k} \not\supseteq \Pi_m^{k+1}$ and $\Pi_m^k \not\supseteq \Sigma_m^{k+1}$. Consequently, we have
$\Sigma^k_m \subsetneq \Sigma_{m+1}^{k+1}$ and $\Pi_m^k \subsetneq \Pi_{m+1}^{k+1}$. These results hold independently of the underlying signature
$\Prop$ and $\Act$.
\end{theorem}


\section{Conclusion and Further Work}
\label{sec:concl}

We have shown that the arity hierarchy in the polyadic $\mu$-calculus, a modal fixpoint logic for specifying
bisimulation-invariant relational properties of states in transition systems, is strict in the sense that
higher arity gives higher expressive power provided that one is allowed to use a little bit more fixpoint
alternation ($\Sigma_m^k \subsetneq \Sigma_{m+1}^{k+1}$). If alternation must not increase then higher arity
yields not necessarily more but different expressiveness ($\Sigma^k_m \not\supseteq \Pi_m^{k+1}$). 

Obviously, the exact effects on expressive power that should be attributed to arity and to fixpoint alternation
need to be separated. A first step would be to prove the strictness of the alternation hierarchy within each
$\mucalc{k}$. For $k=1$, i.e.\ the ordinary $\mu$-calculus, this is known for arbitrary and in particular
for finite transition systems \cite{CONCUR::Bradfield1996,ICALP::Lenzi1996}. Subsequently, the result could be shown for several
other classes of transition systems, for instance binary trees \cite{RAIRO:arnold99,bradfield:rairo99}, nested words \cite{GutierrezKL14}
and graphs whose edge relation satisfies certain properties like being transitive for instance \cite{DAgostino_Lenzi2010:transitive,Alberucci_Facchini2009:transitive_reflexive}.

We suspect that not only is the alternation hierarchy within each \mucalc{k} also strict, but equally that Arnold's
proof \cite{RAIRO:arnold99} using a similar diagonalisation argument for \mucalc{} can be extended. It relies on 
the interreducibility between model checking for \mucalc{} and parity games \cite{Stirling95} and in particular the
existence of the Walukiewicz formulas defining winning regions in parity games \cite{journals/tcs/Walukiewicz02}. It is
known \cite{LL-FICS12} that the model checking problem for \mucalc{k} and any $k \ge 1$ can equally be reduced to a parity
game, and it seems feasible to extend the construction of the Walukiewicz formulas to higher arity. This would use
similar principles as those underlying the construction of $\Phi_m^{k+1}$ in Section~\ref{sec:hierarchy}.

Model checking \mucalc{k} can also be reduced to model checking \mucalc{} directly using $k$-products of transition
systems, i.e.\ there is a translation of \mucalc{k}-formulas to \mucalc{}-formulas that preserves truth under
taking $k$-fold products of transition systems \cite{Otto/99b,LL-FICS12}. Hence, the question of the strictness
of the alternation hierarchy in \mucalc{k} is equivalent to the question after the strictness of the \mucalc{} alternation
hierarchy over the class of all $k$-fold products of transition systems.



\end{document}